%% file: Penalty_shootout_design.tex
\documentclass[12pt]{article}
\usepackage[latin1]{inputenc}
\usepackage[british]{babel}
\usepackage{cmap}
\usepackage{lmodern}

\usepackage{amssymb, amsmath, amsthm}
\usepackage[a4paper,top=25mm,bottom=25mm,left=25mm,right=25mm]{geometry}
\usepackage{etex}
\usepackage{ragged2e}

\usepackage{authblk} 
\usepackage{pifont}
\usepackage{graphicx}
\usepackage[usenames,dvipsnames,svgnames,table]{xcolor}
\usepackage[figuresright]{rotating}
\usepackage{xtab} 
\usepackage{longtable} 
\usepackage{multirow}
\usepackage{footnote}
\usepackage[stable]{footmisc}
\usepackage{chngpage} 
\usepackage{pdflscape} 
\usepackage[nottoc,notlot,notlof]{tocbibind} 

\usepackage{pgfplots}
\pgfplotsset{compat=1.14}
\pgfplotsset{every tick label/.append style={font=\footnotesize}}
\usepgfplotslibrary{fillbetween}
\usepackage{setspace}

\makesavenoteenv{tabular}
\usepackage{tabularx}
\usepackage{booktabs}
\usepackage[flushleft]{threeparttable}
\usepackage[referable]{threeparttablex} 
\newcolumntype{R}{>{\raggedleft\arraybackslash}X}
\newcolumntype{L}{>{\raggedright\arraybackslash}X}
\newcolumntype{C}{>{\centering\arraybackslash}X}
\newcolumntype{M}[1]{>{\centering\arraybackslash}m{#1}}
\newcolumntype{K}{>{\columncolor{gray!20}}C}
\newcolumntype{k}{>{\columncolor{gray!20}}c}

\newlength{\tablen}

\usepackage{dcolumn} 
\newcolumntype{.}{D{.}{.}{-1}}

\usepackage{tikz}
\usetikzlibrary{arrows, calc, matrix, patterns, positioning, trees}
\usepackage[semicolon]{natbib}
\usepackage[hyphens]{url}
\usepackage{hyperref} 
\hypersetup{
  colorlinks   = true,    
  urlcolor     = blue,    
  linkcolor    = blue,    
  citecolor    = ForestGreen      
}
\usepackage{microtype}
\usepackage[justification=centering]{caption} 

\usepackage[labelformat=simple]{subcaption}

\DeclareCaptionLabelFormat{parenthesis}{(#2)}
\makeatletter
\renewcommand\p@subfigure{\arabic{figure}.}
\makeatother

\DeclareCaptionLabelFormat{parenthesis}{(#2)}
\makeatletter
\renewcommand\p@subtable{\arabic{table}.}
\makeatother

\usepackage{enumitem}
\setlist[itemize]{leftmargin=2.5\parindent}
\setlist[enumerate]{leftmargin=2.5\parindent}

\newenvironment{customlegend}[1][]{%
	\begingroup
	\csname pgfplots@init@cleared@structures\endcsname
	\pgfplotsset{#1}%
    }{%
	\csname pgfplots@createlegend\endcsname
	\endgroup
    }%
\def\addlegendimage{\csname pgfplots@addlegendimage\endcsname}

\theoremstyle{plain}

\newtheorem{proposition}{Proposition}

\theoremstyle{definition}

\newtheorem{definition}{Definition}

\theoremstyle{remark}

\def\keywords{\vspace{.5em} 
{\noindent \textit{Keywords}: }}

\def\JEL{\vspace{.5em} 
{\noindent \textbf{\emph{JEL} classification number}: }}

\def\AMS{\vspace{.5em} 
{\noindent \textbf{\emph{MSC} class}: }}

\author{\href{https://sites.google.com/view/laszlocsato}{L\'aszl\'o Csat\'o}\thanks{~E-mail: \emph{laszlo.csato@sztaki.hu}} }
\affil{Institute for Computer Science and Control (SZTAKI) \\
Laboratory on Engineering and Management Intelligence, Research Group of Operations Research and Decision Systems}
\affil{Corvinus University of Budapest (BCE) \\
Department of Operations Research and Actuarial Sciences}
\affil{Budapest, Hungary}
\title{A comparison of penalty shootout designs in soccer}
\date{\today}

\def\Dedication{ 
{\noindent $\mathfrak{Alles}$ $\mathfrak{erscheint}$ $\mathfrak{so}$ $\mathfrak{einfach}$, $\mathfrak{alle}$ $\mathfrak{erforderlichen}$ $\mathfrak{Kenntnisse}$ $\mathfrak{erscheinen}$ $\mathfrak{so}$ $\mathfrak{flach}$, $\mathfrak{alle}$ \linebreak $\mathfrak{Kombinationen}$ $\mathfrak{so}$ $\mathfrak{unbedeutend}$, $\mathfrak{daß}$ $\mathfrak{in}$ $\mathfrak{Vergleichung}$ $\mathfrak{damit}$ $\mathfrak{uns}$ $\mathfrak{die}$ $\mathfrak{einfachste}$ $\mathfrak{Aufgabe}$ $\mathfrak{der}$ $\mathfrak{h\ddot{o}heren}$ $\mathfrak{Mathematik}$ $\mathfrak{mit}$ $\mathfrak{einer}$ $\mathfrak{gewissen}$ $\mathfrak{wissenschaftlichen}$ $\mathfrak{W\ddot{u}rde}$ $\mathfrak{imponiert}$.}\footnote{~``\emph{All appears so simple, all the requisite branches of knowledge appear so plain, all the combinations so unimportant, that, in comparison with them, the easiest problem in higher mathematics impresses us with a certain scientific dignity.}'' (Source: Carl von Clausewitz: \emph{On War}, Book 1, Chapter 7 -- Friction in War, translated by Colonel James John Graham, London, N. Tr\"ubner, 1873. \url{http://clausewitz.com/readings/OnWar1873/TOC.htm})}
\vspace{0.25cm}

\flushright
\noindent (Carl von Clausewitz: \emph{Vom Kriege})

\vspace{1cm} 
\justify }

\begin{document}

\maketitle

\Dedication

\begin{abstract}
\noindent
Penalty shootout in soccer is recognized to be unfair because the team kicking first in all rounds enjoys a significant advantage.
The so-called Catch-Up Rule has been suggested recently to solve this problem but is shown here not to be fairer than the simpler deterministic Alternating (ABBA) Rule that has already been tried. We introduce the Adjusted Catch-Up Rule by guaranteeing the first penalty of the possible sudden death stage to the team disadvantaged in the first round. It outperforms the Catch-Up and Alternating Rules, while remains straightforward to implement.
A general measure of complexity for penalty shootout mechanisms is also provided as the minimal number of binary questions required to decide the first-mover in a given round without knowing the history of the penalty shootout. This quantification permits a two-dimensional evaluation of any mechanism proposed in the future.


\AMS{60J20, 91A05, 91A80}

\JEL{C44, C72, Z20}
\end{abstract}

\section{Introduction} \label{Sec1}

Fairness has several interpretations in sports, one basic desideratum being the interpretation of the Aristotelian Justice principle: higher-ability competitors should win with a higher probability alongside the equal treatment of equals.
In particular, we address the problem of penalty shootouts in soccer (association football) from this point of view.

According to the current rulebook of soccer, Laws of the Game 2019/20, ``\emph{when competition rules require a winning team after a drawn match or home-and-away tie, the only permitted procedures to determine the winning team are:
a) away goals rule;
b) two equal periods of extra time not exceeding 15 minutes each;
c) kicks from the penalty mark}'' \citep[Section~10]{IFAB2019}.
In the ultimate case of item \emph{c)}, a coin is tossed to decide the goal at which the kicks will be taken. Then the referee tosses a coin again, the winner decides whether to take the first or second kick, and five kicks are taken alternately by both teams (if, before both teams have taken five kicks, one has scored more goals than the other could score, even if it were to complete its five kicks, no more kicks are taken).
If the scores are still level after five rounds, the kicks continue in the \emph{sudden death} stage until one team scores a goal more than the other from the same number of kicks. Following \citet{BramsIsmail2018}, we will refer to this rule as the \emph{Standard ($ABAB$) Rule}.

Since most penalties are successful in soccer, the player taking the second kick is usually under greater mental pressure, especially from the third or fourth penalties onward, when a miss probably means the loss of the match. Consequently, the team kicking first in a penalty shootout is recognized to win significantly more frequently than 50 percent of the time \citep{ApesteguiaPalacios-Huerta2010, Palacios-Huerta2014, DaSilvaMioranzaMatsushita2018, RudiOlivaresShatty2019}, indicating the unfairness of the Standard ($ABAB$) Rule.

Therefore, three alternative mechanisms for penalty shootouts will be considered here:
\begin{itemize}
\item
\emph{Alternating ($ABBA$) Rule}: the order of the first two penalties ($AB$) is mirrored in the next two ($BA$), and this sequence continues even in the possible sudden death stage of the shootout (the sixth round of penalties is started by team $B$, the seventh by team $A$, and so on).
\item
\emph{Catch-Up Rule} \citep{BramsIsmail2018}: the order of the penalties in a given round, including the sudden death, is the mirror image of the previous round except if the first team failed and the second scored in the previous round when the order of the teams remains unchanged.
\item
\emph{Adjusted Catch-Up Rule}: the first five rounds of penalties, started by team $A$, are kicked according to the Catch-Up Rule, but team $B$ is guaranteed to be the first kicker in the sudden death stage (sixth round).
\end{itemize}
Note that the Adjusted Catch-Up Rule combines the other two mechanisms: it coincides with the Catch-Up Rule in the first five rounds and with the Alternating ($ABBA$) Rule in the sudden death stage.

The three designs will be compared concerning not only fairness but also simplicity because achieving fairness has a price in increasing complexity.

The contribution of our research resides in the following points:
\begin{enumerate}
\item
We find that the Catch-Up Rule, promoted by \citet{BramsIsmail2018}, does not outperform the simpler and already tried Alternating ($ABBA$) Rule under the assumptions of the same authors, which substantially reduces the importance of a central result of \citet{BramsIsmail2018};
\item
We show that the proposed Adjusted Catch-Up Rule is fairer than both alternative penalty shootout designs;
\item
We suggest the first general complexity measure of penalty shootout mechanisms in the literature that remains consistent with the view of the decision-makers. It is based on the minimal number of binary questions required to decide the first-mover in a given round of the penalty shootout without knowing its history.
\end{enumerate}

The paper is organized as follows.
Section~\ref{Sec2} discusses the problem of penalty shootouts, while Section~\ref{Sec3} analyzes the fairness of the three penalty shootout designs presented above.
Quantification of the complexity of an arbitrary rule is provided in Section~\ref{Sec4}.
Section~\ref{Sec5} offers some concluding thoughts.

\section{An overview of penalty shootouts} \label{Sec2}

Soccer is typically a game with a low number of scores, hence ties, even the result of 0-0, are relatively common. Since in knockout (elimination) tournaments only one team advances to the next round, these ties should be broken.

Before 1970, soccer matches that were tied after extra time were either decided by a coin toss or replayed. However, events in the 1968 European football championship led FIFA, the international governing body of association football, to try penalty shootouts \citep{AnbarciSunUnver2019}.\footnote{~The tied semifinal (after extra time) between Italy and the Soviet Union was decided by a coin toss for Italy. The final between Italy and Yugoslavia ended in a draw of 1-1 even after 30 minutes extra time, thus it was replayed two days later.}
In the following decades, penalty shootout has become the standard tie-breaking procedure in knockout tournaments.

Besides, penalty shootout may be a specific tie-breaking rule in round-robin tournaments.
For example, in the group stage of the \href{https://en.wikipedia.org/wiki/UEFA_Euro_2020}{2020 UEFA European Football Championship}, if two teams, which have the same number of points and the same number of goals scored and conceded, play their last group match against each other and are still equal at the end of that match, their final rankings are determined by kicks from the penalty mark, provided that no other teams within the group have the same number of points on completion of all group matches \citep[Article~20.02]{UEFA2018c}.\footnote{~Despite the restrictive conditions, it remains not only a theoretical possibility: in the \href{https://en.wikipedia.org/wiki/2016_UEFA_European_Under-17_Championship_qualification\#Elite_round}{elite round of 2016 UEFA European Under-17 Championship qualification}, Poland overtook Ireland in Group 7, and Belgium was ranked higher than Spain in Group 8 due to this particular rule.}

In the \href{https://en.wikipedia.org/wiki/1988\%E2\%80\%9389_in_Argentine_football}{1988-89 season of the Argentinian League}, all drawn matches went to penalties without extra time, when the winner of the shootout obtained two points, and the loser one point \citep[Section~10]{Palacios-Huerta2014}. A similar rule was applied in the \href{https://en.wikipedia.org/wiki/1994\%E2\%80\%9395_National_Soccer_League}{1994-95 Australian National Soccer League}, except that a regular win was awarded by four points \citep[Section~3.9.7]{KendallLenten2017}.

\subsection{On the fairness of penalty shootouts} \label{Sec21}

Penalty shootouts have inspired many academic researchers to investigate the issue of fairness as they offer excellent natural experiments despite a substantial rule change in their implementation: before June 2003, the team that won the random coin toss had to take the first kick, and after July 2003, the winner of the coin toss can choose the order of kicking.
The stakeholders also feel potential problems, more than 90\% of coaches and players asked in a survey want to go first, mainly because they attempt to put psychological pressure on the other team \citep{ApesteguiaPalacios-Huerta2010}.\footnote{~An interesting exception was a quarterfinal of the 2018 FIFA World Cup when the Croatian team captain \emph{Luka Modri\'c} chose to kick the second penalties despite winning the coin toss against Russia \citep{Mirror2018}.}

Denote the two teams by $A$ and $B$, where $A$ is the first kicker.
According to \citet{ApesteguiaPalacios-Huerta2010}, team $A$ wins with the probability of 60.5\% based on 129 pre-2003 penalty shootouts and with the probability of 59.2\% based on 269 shootouts that include post-2003 cases. The advantage of the first-mover is statistically significant.
However, using a superset of their pre-2003 sample with 540 shootouts, \citet{KocherLenzSutter2012} report this value to be only 53.3\% and insignificant.
\citet{Palacios-Huerta2014} further expands the database to 1001 penalty shootouts played before 2012 to get a 60.6\% winning probability for the first team.
\citet{VandebroekMcCannVroom2018} explain this disagreement with insufficient sample sizes, they find that even a relatively small but meaningful lagging-behind effect (the team having less score succeeds with only a 70\% probability instead of 75\%) cannot be reliably identified if only 500 penalty shootouts are considered.

\citet{DaSilvaMioranzaMatsushita2018} collect 232 penalty shootout situations and get a 59.48\% winning probability for team $A$, which is statistically significant. On the other hand, \citet{ArrondelDuhautoisLaslier2019} show no advantage based on 252 French penalty shootouts. However, their results reveal that the probability of scoring is negatively affected by the stake (the impact of my scoring on the expected probability that my team will eventually win) and the difficulty of the situation (the ex-ante probability of my team eventually losing).
Finally, \citet{RudiOlivaresShatty2019} investigate 1635 penalty shootouts, which leads to a statistically significant 54.86\% winning probability for team $A$. Although this is closer in magnitude to the value presented by \citet{KocherLenzSutter2012} than to the findings of \citet{ApesteguiaPalacios-Huerta2010} and \citet{Palacios-Huerta2014}, the larger sample size enables a more precise estimation and higher statistical power to detect the possible advantage.

Similar problems may arise in other sports such as handball, ice hockey, or water polo \citep{AnbarciSunUnver2019}.
\citet{Cohen-ZadaKrumerShapir2018} and \citet{DaSilvaMioranzaMatsushita2018} find that the $ABBA$ pattern does not favor the player who serves first in a tennis tiebreak.
According to \citet{Gonzalez-DiazPalacios-Huerta2016}, the player drawing the white pieces in the odd games of a multi-stage chess contest has about 60\% chance to win the match. Therefore, since the \href{https://en.wikipedia.org/wiki/World_Chess_Championship_2006}{World Chess Championship 2006}, the colors are reversed halfway through in the match containing twelve scheduled games as one player plays with the white pieces in the 1st, 3rd, 5th, 8th, 10th, 12th games according to the $ABABAB|BABABA$ sequence.

To summarize, while the empirical evidence remains somewhat controversial, it seems probable that the team kicking the first penalty enjoys an advantage, which is widely regarded as unfair.
This fact is also recognized by the IFAB (International Football Association Board), the rule making body of soccer: Laws of the Game 2017/18 explicitly says in its section discussing the future that the IFAB will consult widely on a number of important Law-related topics, including ``\emph{a potentially fairer system of taking kicks from the penalty mark}'' \citep{IFAB2017}.\footnote{~This sentence appears in the same place in Laws of the Game 2018/19 \citep{IFAB2018}, but is missing from Laws of the Game 2019/20 \citep{IFAB2019}.}

\subsection{Alternative mechanisms for penalty shootouts} \label{Sec22}

The IFAB has decided to test the \emph{Alternating ($ABBA$) Rule}.
The trial was initially scheduled at the \href{https://en.wikipedia.org/wiki/2017_UEFA_European_Under-17_Championship}{2017 UEFA European Under-17 Championship} and the \href{https://en.wikipedia.org/wiki/2017_UEFA_Women\%27s_Under-17_Championship}{2017 UEFA Women's Under-17 Championship}, organized in May 2017 \citep{UEFA2017e}, and was extended to the \href{https://en.wikipedia.org/wiki/2017_UEFA_European_Under-19_Championship}{2017 UEFA European Under-19 Championship} and the \href{https://en.wikipedia.org/wiki/2017_UEFA_Women's_Under-19_Championship}{2017 UEFA Women's Under-19 Championship} in the following month \citep{UEFA2017f}. The first implementation of the new system was a penalty shootout between Germany and Norway in the Women's Under-17 Championship semifinal on 11 May 2017 \citep{ThomsonReuters2017}.

This mechanism was applied in the \href{https://en.wikipedia.org/wiki/2017_FA_Community_Shield}{2017 FA Community Shield}, too, where Arsenal, the winner of the 2017 FA Cup Final, won after an $ABBA$ penalty shootout against Chelsea, the champions of the 2016/17 Premier League.
There was even a controversy in the Dutch KNVB Cup in 2017 when a referee erroneously employed the Alternating ($ABBA$) rule during a penalty shootout, hence it should be replayed three weeks after \citep{Mirror2017}.

However, the 133rd Annual Business Meeting (ABM) of the IFAB agreed that the Alternating ($ABBA$) rule will no longer be a future option for competitions due to ``\emph{the absence of strong support, mainly because the procedure is complex}'' \citep{FIFA2018b}.

Academic researchers have proposed some further rules to increase fairness \citep{AnbarciSunUnver2019, BramsIsmail2018, Echenique2017, Palacios-Huerta2012}.
Our point of departure is the \emph{Catch-Up Rule} \citep{BramsIsmail2018}, which takes into account the results of penalties in the preceding round to allow the team performing worse to catch up.
Assume that team $A$ kicks first in a particular round, thus it is advantaged. In the next round, team $B$ will kick first except if $A$ fails and $B$ succeeds.

We suggest a slight improvement in this mechanism. Note that the penalty shootout is essentially composed of two parts, the first five rounds, and the possible sudden death stage. Therefore, it makes sense to balance the advantage of the first-mover by making it disadvantaged at the beginning of the sudden death. Formally, if team $A$ starts the shootout, then team $B$ will kick first in the sixth round, provided that it is reached. Under the original Catch-Up Rule, it is possible that $A$ kicks first in the sixth round, for instance, when it leads by 4-3 after four rounds, but $A$ fails and $B$ succeeds in the fifth round of penalty kicks.
This variant of the Catch-Up Rule, which \emph{a priori} fixes the first-mover in the sudden death, is called the \emph{Adjusted Catch-Up Rule}.



\begin{table}[ht!]
\centering
\caption{An example of penalty shootout rules}
\label{Table1}
	\begin{tabularx}{\textwidth}{l KKCC KKCC} \toprule 
    Rule & \multicolumn{2}{k}{$ABAB$} & \multicolumn{2}{c}{$ABBA$} & \multicolumn{2}{k}{Catch-Up} & \multicolumn{2}{c}{Adj. Catch-Up} \\ \midrule
    Team & \textcolor{red}{Red} & \textcolor{blue}{Blue} & \textcolor{red}{Red} & \textcolor{blue}{Blue} & \textcolor{red}{Red} & \textcolor{blue}{Blue} & \textcolor{red}{Red} & \textcolor{blue}{Blue} \\ \midrule \midrule 
    1st kick & \textcolor{red}{\ding{52}}     &       & \textcolor{red}{\ding{52}}     &       & \textcolor{red}{\ding{52}}     &       & \textcolor{red}{\ding{52}}     &  \\
    2nd  &       & \textcolor{blue}{\ding{52}}     &       & \textcolor{blue}{\ding{52}}     &       & \textcolor{blue}{\ding{52}}     &       & \textcolor{blue}{\ding{52}} \\ \midrule
    3rd  & \textcolor{red}{\ding{55}}     &       &     & \textcolor{blue}{\ding{55}} &      & \textcolor{blue}{\ding{55}}  &     & \textcolor{blue}{\ding{55}} \\
    4th  &       & \textcolor{blue}{\ding{55}}     & \textcolor{red}{\ding{55}} &      & \textcolor{red}{\ding{55}}   &    & \textcolor{red}{\ding{55}}   &  \\ \midrule
    5th  & \textcolor{red}{\ding{52}}     &       & \textcolor{red}{\ding{52}}     &       & \textcolor{red}{\ding{52}}     &       & \textcolor{red}{\ding{52}}     &  \\
    6th  &       & \textcolor{blue}{\ding{52}}     &       & \textcolor{blue}{\ding{52}}     &       & \textcolor{blue}{\ding{52}}     &       & \textcolor{blue}{\ding{52}} \\ \midrule
    7th  & \textcolor{red}{\ding{52}}     &       &       & \textcolor{blue}{\ding{55}}      &       & \textcolor{blue}{\ding{55}}      &       & \textcolor{blue}{\ding{55}} \\
    8th  &       & \textcolor{blue}{\ding{55}}      & \textcolor{red}{\ding{52}}     &       & \textcolor{red}{\ding{52}}     &       & \textcolor{red}{\ding{52}}     &  \\ \midrule
    9th  & \textcolor{red}{\ding{55}}      &       & \textcolor{red}{\ding{55}}      &       &       & \textcolor{blue}{\ding{52}}     &       & \textcolor{blue}{\ding{52}} \\
    10th &       & \textcolor{blue}{\ding{52}}   &       & \textcolor{blue}{\ding{52}}   & \textcolor{red}{\ding{55}}    &       & \textcolor{red}{\ding{55}}  &  \\ \midrule \midrule
    11th & \textcolor{red}{\ding{52}}     &       &       & \textcolor{blue}{\ding{52}}     & \textcolor{red}{\ding{52}}     &       &       & \textcolor{blue}{\ding{52}} \\
    12th &       & \textcolor{blue}{\ding{52}}     & \textcolor{red}{\ding{52}}     &       &       & \textcolor{blue}{\ding{52}}     & \textcolor{red}{\ding{52}}     &  \\ \midrule
    13th & \textcolor{red}{\ding{52}}     &       & \textcolor{red}{\ding{52}}      &     &     & \textcolor{blue}{\ding{55}}      & \textcolor{red}{\ding{52}}      &  \\
    14th &       & \textcolor{blue}{\ding{55}}   &     & \textcolor{blue}{\ding{55}}      & \textcolor{red}{\ding{52}}     &     &      & \textcolor{blue}{\ding{55}} \\ \bottomrule
	\end{tabularx}
\end{table}

Table~\ref{Table1} illustrates how the four rules work. The Red team is the first kicker, \ding{52} means a successful, and \ding{55} indicates an unsuccessful penalty. Since the result after five rounds is 3-3, the sudden death stage starts: the Red team kicks first in the sixth round according to the Catch-Up Rule as the Blue team was the first-mover in the previous round, but the Blue team kicks first in the sixth round when the Adjusted Catch-Up Rule is used because it was disadvantaged in the first round.

\section{The analysis of three penalty shootout designs} \label{Sec3}

Following the literature on penalty shootouts, fairness means that no team enjoys an advantage because of winning or losing the coin toss.

\begin{definition} \label{Def31}
\emph{Fairness}:
A penalty shootout mechanism is \emph{fair} if the probability of winning does not depend on the outcome of the coin toss.
\end{definition}

Consequently, in our mathematical model, a mechanism is called \emph{fairer} than another if the probability of winning the match conditional on winning the coin toss is closer to $0.5$ for equally skilled teams.

The standard $ABAB$ rule will not be discussed here because it has already been investigated in \citet{BramsIsmail2018} --- and is markedly unfair.

\subsection{Fairness: a simple model which solely depends on the order} \label{Sec31}

\begin{table}[ht!]
\centering
\caption{Penalty shootout success rates per round}
\label{Table2}
\rowcolors{3}{gray!20}{}
\begin{threeparttable}
    \begin{tabularx}{0.6\textwidth}{lCC} \toprule \hiderowcolors
          & First kicker & Second kicker \\ \midrule \showrowcolors
    Round 1 & 0.79  & 0.72 \\
    Round 2 & 0.82  & 0.77 \\
    Round 3 & 0.77  & 0.64 \\
    Round 4 & 0.74  & 0.68 \\
    Round 5 & 0.74  & 0.67 \\ \bottomrule
    \end{tabularx}
\vspace{0.25cm}
\begin{tablenotes} \footnotesize
\item
Source: \citet[p.~2558]{ApesteguiaPalacios-Huerta2010}
\end{tablenotes}
\end{threeparttable}
\end{table}

\citet[p.~2558]{ApesteguiaPalacios-Huerta2010} provide empirical probabilities for scoring a penalty on each round, presented in Table~\ref{Table2}.
It can be seen that the team kicking first in a given round always succeeds with a higher probability. Hence, following \citet{BramsIsmail2018}, we use the reasonable assumption that the probability of a successful penalty depends only on whether the team kicks first or second in a round: the advantaged team has a probability $p$ of scoring, and the disadvantaged team has a probability $q (\leq p)$ of scoring. Similarly to \citet{BramsIsmail2018}, our baseline choice is $p=3/4$ and $q=2/3$, which are close to the empirical success rates given in Table~\ref{Table2}, especially in the last three rounds. This corresponds to about a 60\% chance of winning for the first-mover as observed in practice by \citet{ApesteguiaPalacios-Huerta2010} and \citet{Palacios-Huerta2014}.

To illustrate the model, \citet{BramsIsmail2018} analyze the Catch-Up Rule for a penalty shootout over only two rounds and derive that $p=3/4$ and $q=2/3$ result in:
\begin{itemize}
\item
the probability of team $A$ winning is $P^2(A) = 41/144 \approx 0.285$;\footnote{~\citet[p.~188]{BramsIsmail2018} contains a rounding error. Superscript $2$ indicates that the probability concerns a penalty shootout over two rounds.}
\item
the probability of team $B$ winning is $P^2(B) = 39/144 \approx 0.270$;
\item
the probability of a tie is $P^2(T) = 64/144 \approx 0.444$.
\end{itemize}

If there is a tie after two rounds, the shootout goes to sudden death. Assume that team $A$ kicks first and let $W(A)$ be the probability of winning for team $A$ in the sudden death stage. The Catch-Up, Adjusted Catch-Up, and Alternating ($ABBA$) Rules coincide in this stage, the calculations of \citet{BramsIsmail2018} remain valid, that is,
\begin{equation} \label{eq1}
W(A) = \frac{1-q+pq}{2-p-q+2pq}.
\end{equation}
For $p=3/4$ and $q=2/3$, one gets $W(A) = 10/19 \approx 0.526$.

If the penalty shootout is played over two rounds before sudden death, the probability of a tie is $P^2(T) = 64/144$.
Under the Catch-Up Rule, $A$ kicks first in the third round with a probability of $58/144 \approx 0.403$, while $B$ kicks first in the third round with a probability of $6/144 \approx 0.042$ because team $B$ will kick first only in the case of the following sequence: $A$ fails, $B$ scores, $A$ scores, $B$ fails, which has a probability of $(1-p)qp(1-q)$. Consequently, the probability that team $A$ wins is
\[
Q^2(A) = P^2(A) + \frac{58}{144} \times \frac{10}{19} + \frac{6}{144} \times \frac{9}{19} = \frac{1413}{2736} \approx 0.516.
\]

On the other hand, the Adjusted Catch-Up Rule guarantees the first penalty in the sudden death for team $B$, hence the probability that team $A$ wins under this mechanism is
\[
Q^2(A) = P^2(A) + \left( \frac{58}{144} + \frac{6}{144} \right) \times \frac{9}{19} = \frac{1355}{2736} \approx 0.495.
\]

A more detailed discussion of the Alternating ($ABBA$) Rule is provided because it is missing from \citet{BramsIsmail2018} but can contribute to a better understanding of the model.
There are three ways for team $A$ to win a penalty shootout over two rounds:
\begin{enumerate}[label=\Roman*)]
\item
2-0: \emph{$A$ scores on both rounds while $B$ fails to score on both rounds} \\
On the first round, $A$ succeeds and $B$ fails with probability $p(1-q)$.
On the second round, $B$ kicks first and fails, while $A$ kicks second and succeeds with probability $(1-p)q$.
The joint probability of this outcome over the two rounds is $p(1-q) (1-p)q$.
\item
2-1: \emph{$A$ scores on both rounds while $B$ fails to score on one of these rounds} \\
There are two subcases:
\begin{itemize}
\item
\emph{$B$ scores on the first round} \\
On this round, both teams succeed with probability $pq$.
On the second round, $B$ kicks first and fails, while $A$ kicks second and scores with probability $(1-p)q$.
The joint probability over both rounds is $pq (1-p)q$.
\item
\emph{$B$ scores on the second round} \\
On the first round, $A$ succeeds and $B$ fails with probability $p(1-q)$.
On the second round, $B$ kicks first and succeeds, after which $A$ also scores, with probability $pq$. The joint probability over both rounds is $p(1-q) pq$.
\end{itemize}
Hence the probability of the outcome 2-1 is
\[
pq (1-p)q + p(1-q) pq.
\]
\item
1-0: \emph{$A$ scores on one round while $B$ fails to score on both rounds} \\
There are two subcases:
\begin{itemize}
\item
\emph{$A$ scores on the first round} \\
On this round, $A$ succeeds and $B$ fails with probability $p(1-q)$.
On the second round, both teams fail with probability $(1-p)(1-q)$.
The joint probability over the two rounds is $p(1-q) (1-p)(1-q)$.
\item
\emph{$A$ scores on the second round} \\
On the first round, both teams fail with probability $(1-p)(1-q)$.
On the second round, $B$ kicks first and fails, after which $A$ succeeds, with probability $(1-p)q$. The joint probability over the two rounds is $(1-p)(1-q) (1-p)q$.
\end{itemize}
Thus the probability of the outcome 1-0 is
\[
p(1-q) (1-p)(1-q) + (1-p)(1-q) (1-p)q.
\]
\end{enumerate}
The assumption $p=3/4$ and $q=2/3$ implies that:
\begin{itemize}
\item
the probability of team $A$ winning is $P^2(A) = 41/144 \approx 0.285$;
\item
the probability of team $B$ winning is $P^2(B) = 41/144 \approx 0.285$;
\item
the probability of a tie is $P^2(T) = 62/144 \approx 0.431$.
\end{itemize}
Unsurprisingly, this rule leads to equal winning probabilities for the two teams over two rounds as two is an even number.

The Alternating ($ABBA$) Rule provides the first penalty in the sudden death for team $A$ because it is the third round, hence the probability that team $A$ wins is
\[
Q^2(A) = P^2(A) + \frac{62}{144} \times \frac{10}{19} = \frac{1399}{2736} \approx 0.511.
\]

To summarize, while all three alternative designs tend to equalize the winning probabilities compared to the Standard ($ABAB$) Rule, the Adjusted Catch-Up Rule seems to be the closest to fairness. In particular, the Catch-Up and Alternating ($ABBA$) Rules give $100 \times (0.516 / 0.484 - 1) \approx 6.8$\%  and $4.64$\% advantage for the team kicking the first penalty, respectively, while the Adjusted Catch-Up Rule results in an advantage of only $1.92$\% for \emph{the other team} in a penalty shootout over two rounds with sudden death.

\begin{table}[ht!]
\centering
\caption{The probability that $A$ wins including sudden death ($p = 3/4$ and $q = 2/3$)}
\label{Table3}
\rowcolors{3}{gray!20}{}
    \begin{tabularx}{\textwidth}{Lccc} \toprule \hiderowcolors
          & Catch-Up Rule & Adjusted Catch-Up Rule & Alternating ($ABBA$) Rule \\ \midrule \showrowcolors
    1 Round  & 0.526 & 0.526 & 0.526 \\
    2 Rounds & 0.516 & 0.495 & 0.511 \\
    3 Rounds & 0.518 & 0.515 & 0.519 \\
    4 Rounds & 0.513 & 0.501 & 0.508 \\
    \textbf{5 Rounds} & \textbf{0.514} & \textbf{0.509} & \textbf{0.515} \\
    6 Rounds & 0.512 & 0.504 & 0.507 \\
    7 Rounds & 0.512 & 0.507 & 0.513 \\
    8 Rounds & 0.511 & 0.504 & 0.506 \\ \hline
    \end{tabularx}
\end{table}

The winning probabilities of the advantaged team, which kicks the first penalty, are shown in Table~\ref{Table3} for penalty shootouts lasting eight or fewer predetermined rounds followed by sudden death when $p = 3/4$ and $q = 2/3$. Note that the probabilities for the Catch-Up Rule have already been reported in \citet{BramsIsmail2018} up to five rounds.

All three methods, especially the Alternating ($ABBA$) Rule, exhibit a small odd-even effect since their bias is greater for an odd number of predetermined rounds. As expected, they make the contest fairer if the number of rounds increases. The simplest Alternating ($ABBA$) Rule is better than the Catch-Up Rule for an even number of rounds, while the latter has a marginal advantage for an odd number of rounds.

However, the Adjusted Catch-Up Rule consistently outperforms both of them. The smallest imbalance can be observed for a penalty shootout played over four rounds, followed by sudden death if the shootout remains unresolved. In this case, the team kicking first has only $0.58$\% more chance to win under the Adjusted Catch-Up Rule.

\input{Figure1_penalty_design_comparison}

Until now, we have investigated only the case of $p = 3/4$ and $q = 2/3$.
Figure~\ref{Fig1} plots the winning probabilities of team $A$ using the presented rules for different values of $p$ as a function of $q$, where $0.5 \leq q \leq p$ since the penalties in soccer are usually successful. It shows that the order of these designs with respect to fairness is not influenced by the particular parameters chosen: the Catch-Up and the Alternating ($ABBA$) Rules remain almost indistinguishable, and the Adjusted Catch-Up Rule turns out to be the best as before. Furthermore, all mechanisms are fairer if $p$ is closer to $q$, according to our intuition.

Unfortunately, there is no hope to analytically derive conditions for $p$ and $q$ which make the Adjusted Catch-Up Rule fairer compared to the other designs even in this simple mathematical model. The reason is that the five rounds of penalties imply $2^{10} = 1024$ different cases, and the probability of each is given by a formula containing the product of ten items from the set of $p$, $q$, $(1-p)$, and $(1-q)$.
Nevertheless, Figure~\ref{Fig1} supports this conjecture by reinforcing the lack of non-linear effects.

\subsection{Fairness: empirical round dependent scoring probabilities} \label{Sec32}

\begin{figure}[ht!]
\centering
\caption{The empirical probability that $A$ wins a penalty \\ shootout over five rounds including sudden death}
\label{Fig2}

\begin{tikzpicture}
\begin{axis}[width=\textwidth, 
height=0.5\textwidth,
tick label style={/pgf/number format/fixed},
symbolic x coords={($2/3;3/5$),($3/4;2/3$),($3/4;3/5$)},
xtick = data,
xlabel = Probabilities of scoring in the sudden death ($p;q$),
enlarge x limits={abs=2cm},
ybar,
ymin = 0.49,
ymax = 0.54,
ymajorgrids = true,
bar width = 1cm,
ybar = 0.25cm,
legend entries = {Catch-Up Rule$\quad$,Adjusted Catch-Up Rule$\quad$,Alternating ($ABBA$) Rule},
legend style = {at={(0.5,-0.25)},anchor = north,legend columns = 3}
]

\addplot [blue, pattern color = blue, pattern = dots, very thick] coordinates{
(($2/3;3/5$),0.52794530463813)
(($3/4;2/3$),0.527520306595316)
(($3/4;3/5$),0.525470719259806)
};

\addplot [red, pattern color = red, pattern = grid, very thick] coordinates{
(($2/3;3/5$),0.523678538465)
(($3/4;2/3$),0.522355273859421)
(($3/4;3/5$),0.515973723584129)
};

\addplot [ForestGreen, pattern color = ForestGreen, pattern = horizontal lines, very thick] coordinates{
(($2/3;3/5$),0.538255247188209)
(($3/4;2/3$),0.536959358460168)
(($3/4;3/5$),0.530709830562039)
};

\draw [ultra thick, dashed] (rel axis cs:0,0.2)  -- (rel axis cs:1,0.2);
\end{axis}
\end{tikzpicture}
\end{figure}

The three rules can also be compared in the view of the empirical round dependent probabilities from Table~\ref{Table2}. Since success rates in the sudden death stage are uncertain due to the small sample size, it is assumed that our former mathematical model holds after five rounds with the fixed probabilities $p$ and $q$.
Figure~\ref{Fig2} presents the results of these calculations. While the Catch-Up Rule is closer to fairness based on the empirical data than the Alternating ($ABBA$) Rule, the Adjusted Catch-Up Rule remains the winner.

We have attempted to determine the scoring probabilities $p$ and $q ( \leq p)$ in the sudden death stage which make the Adjusted Catch-Up Rule fairer than the other two mechanisms. Formally, suppose that the following values are known:
\begin{itemize}
\item
$P^5(A)$: the probability that $A$ wins a penalty shootout over five rounds without sudden death under the Catch-Up Rule;
\item
$P^5_A(T)$: the probability that a penalty shootout over five rounds is tied under the Catch-Up Rule and $A$ kicks the sixth penalty according to the Catch-Up Rule;
\item
$P^5_B(T)$: the probability that a penalty shootout over five rounds is tied under the Catch-Up Rule and $B$ kicks the sixth penalty according to the Catch-Up Rule.
\end{itemize}
Furthermore, denote by $\alpha$ the probability of winning the sudden death by the team that kicks first in this stage. Formula \eqref{eq1} implies $0.5 \leq \alpha$ because of the assumption $q \leq p$ to incorporate the psychological effect, which is probably even stronger in the sudden death.

Then the overall probability of winning for team $A$ under the Catch-Up Rule is
\begin{equation} \label{eq2}
P^5(A) + P^5_A(T) \times \alpha + P^5_B(T) \times \left( 1 - \alpha \right),
\end{equation}
while the overall probability of winning for team $A$ under the Adjusted Catch-Up Rule is
\begin{equation} \label{eq3}
P^5(A) + \left( P^5_A(T) + P^5_B(T) \right) \times \left( 1 - \alpha \right).
\end{equation}
The Adjusted Catch-Up Rule is fairer than the Catch-Up Rule if the value of \eqref{eq3} is closer to $0.5$ than the value of \eqref{eq2}.
By using the round dependent empirical scoring probabilities of Table~\ref{Table2}, this results in $0.5 \leq \alpha \leq \alpha(CU) \approx 0.6569$. Thus the Adjusted Catch-Up Rule becomes fairer than the Catch-Up Rule if
\[
0.5 \leq \frac{1-q+pq}{2-p-q+2pq} \leq \alpha(CU) \iff \frac{1 - 2 \alpha(CU) + \alpha(CU) p}{1 - p - \alpha + 2 \alpha(CU) p} \leq q \leq p.
\]
An analogous calculation leads to the conclusion that the Adjusted Catch-Up Rule is fairer than the $ABBA$ Rule if $0.5 \leq \alpha \leq \alpha(ABBA) \approx 0.6252$.

\input{Figure3_fairness_condition}

The range of values $(p;q)$, $q \leq p$ for the scoring probabilities in sudden death that makes the Adjusted Catch-Up Rule fairer than the other two penalty shootout designs with the empirical results from Table~\ref{Table2} are plotted in Figure~\ref{Fig3}. Our proposal outperforms the Catch-Up Rule in the region indicated by the blue vertical lines, while it is preferred to the Alternating ($ABBA$) Rule in the region indicated by the green horizontal lines (the latter is a subset of the former). Since any reasonable value of $q$ lies between these bounds, the Adjusted Catch-Up Rule is the closest to fairness among the three designs with the empirical round dependent success rates of \citet{ApesteguiaPalacios-Huerta2010}.

\subsection{Beyond fairness: expected length and strategy-proofness} \label{Sec33}

\begin{figure}[ht!]
\centering
\caption{The probability that a penalty shootout over five rounds goes to sudden death}
\label{Fig4}

\begin{tikzpicture}
\begin{axis}[width=\textwidth, 
height=0.5\textwidth,
tick label style={/pgf/number format/fixed},
symbolic x coords={($2/3;3/5$),($3/4;2/3$),($3/4;3/5$),Empirical},
xtick = data,
xlabel = Values of ($p;q$),
enlarge x limits={abs=1.5cm},
ybar,
ymin = 0,
ymajorgrids = true,
bar width = 1cm,
ybar = 0.25cm,
legend entries = {(Adjusted) Catch-Up Rule$\quad$,Alternating ($ABBA$) Rule},
legend style = {at={(0.5,-0.25)},anchor = north,legend columns = 6}
]

\addplot [blue, pattern color = blue, pattern = dots, very thick] coordinates{
(($2/3;3/5$),0.264607078189)
(($3/4;2/3$),0.283733603395)
(($3/4;3/5$),0.2809675)
(Empirical,0.289133316319)
};

\addplot [ForestGreen, pattern color = ForestGreen, pattern = horizontal lines, very thick] coordinates{
(($2/3;3/5$),0.256832263375)
(($3/4;2/3$),0.274731545782)
(($3/4;3/5$),0.266798125)
(Empirical,0.2831516870768)
};
\end{axis}
\end{tikzpicture}
\end{figure}

In the model above, the expected length of the sudden death stage is $1 / (p + q -2pq)$, the same for all mechanisms \citep{BramsIsmail2018}. The Catch-Up and Adjusted Catch-Up Rules differ only in which team kicks the first penalty of the sudden death.
However, the probability of reaching this stage is greater with the (Adjusted) Catch-Up Rule than with the Alternating ($ABBA$) Rule as Figure~\ref{Fig4} illustrates based on some particular values of $p$ and $q$, as well as the empirical round dependent success rates given in Table~\ref{Table2}. Consequently, the former mechanisms can make the penalty shootout somewhat more exciting.

It has been presented recently that certain sports rules do not satisfy incentive compatibility, that is, a team might be strictly better off by exerting a lower effort \citep{Csato2018b, Csato2019b, DagaevSonin2018, KendallLenten2017, Vong2017}.
The Alternating ($ABBA$) Rule is not vulnerable to any kind of strategic manipulation since neither team can influence the order of shooting. According to \citet{BramsIsmail2018}, no team is interested in missing a kick under the Catch-Up Rule if $(p-q) \leq 1/2$, which seems likely to be met in practice. The Adjusted Catch-Up Rule offers fewer opportunities to change the order of the penalties since the first-mover in sudden death is fixed, therefore it also satisfies strategy-proofness if the condition $(p-q) \leq 1/2$ holds.

\section{The complexity of penalty shootout designs} \label{Sec4}

Since the IFAB has stopped the trials of the Alternating ($ABBA$) Rule due to its complexity, this should be another important feature of mechanisms for penalty shootouts.
The first attempt to quantify their simplicity has been provided in \citet{AnbarciSunUnver2015}, and has been repeated in \citet{AnbarciSunUnver2019}. They call a rule simple if it has a stationary machine representation with only two states such that in one state team $A$ kicks first and in the other team $B$ kicks first. However, this measure is not consistent with the decision of the IFAB \citep{FIFA2018b} because it judges the Standard ($ABAB$) and the Alternating ($ABBA$) Rules to have the same level of complexity.

\citet{RudiOlivaresShatty2019} suggest another measure of simplicity but they choose complexity levels somewhat arbitrarily and their approach is not able to classify stochastic mechanisms (such as the Catch-Up Rule), which depend on the outcome of previous penalties.

Thus we provide a procedure that quantifies the complexity of any penalty shootout design, remains intuitive, and is consistent with the recent decision of the IFAB.

\begin{definition}  \label{Def21}
\emph{Complexity}:
Suppose that the \emph{mathematician} should report the \emph{referee} on which team is the first-mover in the next round of a penalty shootout. The mathematician has initially no information but she can ask binary questions on the history of the shootout including the number of the next round.
The \emph{complexity} of any penalty shootout mechanism is the minimal number of questions needed to determine the first kicking team in a given round, taking into account that the questions and their number might depend on the answer(s) to the preceding question(s).
\end{definition}

In other words, there is an information asymmetry between the mathematician and the referee as the former knows only the rules, while the latter knows only the history of the shootout.

Definition~\ref{Def21} can be applied to reveal the simplicity of a penalty shootout mechanism.

\begin{proposition} \label{Prop21}
The complexities of some penalty shootout designs are as follows:
\begin{itemize}
\item
Standard $(ABAB$) Rule: $0$;
\item
Alternating $(ABBA$) Rule: $1$;
\item
Catch-Up Rule: $2$;
\item
Adjusted Catch-Up Rule: between $2$ and $3$.
\end{itemize}
\end{proposition}

\begin{proof}
According to the Standard ($ABAB$) Rule, team $A$ will be the first-mover in the next round of penalties, which is known without asking any question. \\
The Alternating $(ABBA$) Rule requires the knowledge of the parity (odd: team $A$, even: team $B$) of the next round's number. \\
The Catch-Up Rule can be implemented by asking two questions because it depends on the first kicker in the previous round and on the fact whether the first kicker has failed but the second has scored in the previous round or not. \\
The Adjusted Catch-Up Rule first requires the knowledge of whether the sudden death stage is reached or not. After that, either the Alternating $(ABBA$) Rule (one question) or the Catch-Up Rule (two questions) is applied.
\end{proof}

Our approach seems to provide reasonable estimates of simplicity.
For example, the design consisting of three rounds of $ABBA$ followed by Catch-Up is between $2$ and $3$: first, the mathematician should know whether the next round is one of the first three or not, and then the appropriate design can be implemented with further one ($ABBA$) or two (Catch-Up) questions.
However, the Adjusted Catch-Up Rule is probably simpler than this artificial mechanism because changing the doctrine at the beginning of the sudden death stage can be considered less costly compared to changing the doctrine after three rounds as the rule of aggregation is modified in the sudden death anyway. Hence the Adjusted Catch-Up Rule can be judged only marginally more complex than the Catch-Up Rule.


The application of a more complex mechanism remains questionable unless it yields meaningful gains in fairness and other aspects. The Catch-Up Rule does not seem to be fairer than the Alternating ($ABBA$) Rule based on Table~\ref{Table3} and Figure~\ref{Fig1}, which reduces the significance of \citet{BramsIsmail2018}'s proposal. On the other hand, our Adjusted Catch-Up Rule dominates both of them except for a small increase in complexity.

\section{Conclusions} \label{Sec5}

Tournament organizers supposedly want to guarantee fairness. However, the standard penalty shootout mechanism in soccer contains a well-known bias favoring the first shooter. This means a problem because an order of actions that provides an ex-post advantage to one team may harm efficiency by decreasing the probability of the stronger team to win.
Consequently, there is little excuse to continue the use of the current rule.

We have demonstrated by a mathematical model that the recently suggested Catch-Up Rule is not worth implementing since it is not fairer than the less complex Alternating ($ABBA$) Rule already tried. On the other hand, the Adjusted Catch-Up Rule can be considered as a promising candidate to make penalty shootouts fairer and even more exciting. Finally, the proposed quantification of complexity permits a two-dimensional evaluation of any mechanism recommended in the future.




\section*{Acknowledgments}
\addcontentsline{toc}{section}{Acknowledgements}
\noindent
This paper could not be prepared without \emph{my father} (also called \emph{L\'aszl\'o Csat\'o}), who has written the code making the necessary computations in Python essentially during a weekend. \\
We are deeply indebted to \emph{Steven J. Brams} and \emph{Mehmet S. Ismail}, whose work was a great source of inspiration. \\
We would like to thank \href{http://www.palacios-huerta.com/}{\emph{Ignacio Palacios-Huerta}} for useful information on penalty shootouts and \href{https://sites.google.com/view/doragretapetroczy}{\emph{D\'ora Gr\'eta Petr\'oczy}} for her beneficial help. \\
Seven anonymous reviewers provided valuable comments and suggestions on earlier drafts. \\
We are grateful to the \href{https://en.wikipedia.org/wiki/Wikipedia_community}{Wikipedia community} for contributing to our research by collecting and structuring invaluable information on the sports tournaments discussed. \\
The research was supported by the MTA Premium Postdoctoral Research Program under grant PPD2019-9/2019.

\bibliographystyle{apalike}
\bibliography{All_references}

\end{document}

%% file: Figure1_penalty_design_comparison.tex
\begin{figure}[ht!]
\centering
\caption{The probability that $A$ wins a shootout over five rounds including sudden death}
\label{Fig1}

\begin{tikzpicture}
\begin{axis}[
name = axis1,
title = {$p = 0.65$},
xlabel = Value of $q$,
width = 0.5\textwidth,
height = 0.4\textwidth,
legend style = {font=\small,at={(0.2,-0.15)},anchor=north west,legend columns=6},
ymajorgrids = true,
xmin = 0.5,
xmax = 0.65,
ymin = 0.498,
ymax = 0.537,
max space between ticks=50,
]      

\addlegendentry{Catch-Up Rule$\quad$}
\addplot [blue, dotted, thick] coordinates {
(0.5,0.525300273437499)
(0.51,0.523562559213668)
(0.52,0.521833617997753)
(0.53,0.520112973238883)
(0.54,0.51840018512185)
(0.55,0.516694853996055)
(0.56,0.514996623939367)
(0.57,0.51330518645503)
(0.58,0.511620284299782)
(0.59,0.509941715441386)
(0.6,0.50826933714379)
(0.61,0.506603070178177)
(0.62,0.504942903158171)
(0.63,0.503288896997526)
(0.64,0.5016411894886)
(0.65,0.499999999999999)
};
\addlegendentry{Adjusted Catch-Up Rule$\quad$}
\addplot [red] coordinates {
(0.5,0.514751738281249)
(0.51,0.513696697826338)
(0.52,0.512658143775894)
(0.53,0.511634529931487)
(0.54,0.51062434710905)
(0.55,0.509626121726755)
(0.56,0.508638414399471)
(0.57,0.507659818539704)
(0.58,0.506688958964948)
(0.59,0.50572449051135)
(0.6,0.504765096653594)
(0.61,0.503809488130937)
(0.62,0.50285640157929)
(0.63,0.501904598169296)
(0.64,0.500952862250274)
(0.65,0.499999999999999)
};
\addlegendentry{Alternating ($ABBA$) Rule}
\addplot [ForestGreen, loosely dashed, thick] coordinates {
(0.5,0.5251110546875)
(0.51,0.52345432316409)
(0.52,0.521799153507241)
(0.53,0.520144615599161)
(0.54,0.51848978255417)
(0.55,0.516833728906147)
(0.56,0.51517552882792)
(0.57,0.513514254382137)
(0.58,0.511848973803196)
(0.59,0.510178749809797)
(0.6,0.508502637947711)
(0.61,0.506819684962334)
(0.62,0.505128927200628)
(0.63,0.503429389042058)
(0.64,0.501720081358108)
(0.65,0.499999999999999)
};
\draw [ultra thick, dashed] (axis cs:0.5,0.5)  -- (axis cs:0.65,0.5);
\legend{}
\end{axis}

\begin{axis}[
at = {(axis1.south east)},
xshift = 0.1\textwidth,
title = {$p = 0.7$},
xlabel = Value of $q$,
width = 0.5\textwidth,
height = 0.4\textwidth,
legend style = {font=\small,at={(0.2,-0.15)},anchor=north west,legend columns=6},
ymajorgrids = true,
xmin = 0.5,
xmax = 0.7,
ymin = 0.498,
ymax = 0.537,
]      

\addlegendentry{Catch-Up Rule$\quad$}
\addplot [blue, dotted, thick] coordinates {
(0.5,0.534426666666667)
(0.51,0.532639204800457)
(0.52,0.530861196930458)
(0.53,0.529092046249439)
(0.54,0.527331185282417)
(0.55,0.525578079517269)
(0.56,0.523832231256477)
(0.57,0.522093183685991)
(0.58,0.52036052515721)
(0.59,0.518633893678204)
(0.6,0.516912981610389)
(0.61,0.515197540566901)
(0.62,0.513487386509044)
(0.63,0.51178240503723)
(0.64,0.510082556872903)
(0.65,0.508387883528044)
(0.66,0.506698513158861)
(0.67,0.505014666600412)
(0.68,0.503336663578914)
(0.69,0.501664929098599)
(0.7,0.499999999999999)
};
\addlegendentry{Adjusted Catch-Up Rule$\quad$}
\addplot [red] coordinates {
(0.5,0.520995416666667)
(0.51,0.519829491478846)
(0.52,0.518684131373713)
(0.53,0.517557837646061)
(0.54,0.516449142197786)
(0.55,0.515356605919407)
(0.56,0.514278817099576)
(0.57,0.513214389862091)
(0.58,0.512161962629887)
(0.59,0.511120196615517)
(0.6,0.510087774337661)
(0.61,0.509063398163157)
(0.62,0.508045788874117)
(0.63,0.507033684259678)
(0.64,0.506025837731918)
(0.65,0.505021016965544)
(0.66,0.504018002560878)
(0.67,0.503015586729771)
(0.68,0.502012572004001)
(0.69,0.501007769965789)
(0.7,0.499999999999999)
};
\addlegendentry{Alternating ($ABBA$) Rule}
\addplot [ForestGreen, loosely dashed, thick] coordinates {
(0.5,0.534019166666667)
(0.51,0.532348500064522)
(0.52,0.530680527432549)
(0.53,0.52901426747192)
(0.54,0.527348745478119)
(0.55,0.525682991185033)
(0.56,0.52401603667503)
(0.57,0.522346914353848)
(0.58,0.520674654989077)
(0.59,0.518998285811096)
(0.6,0.517316828675324)
(0.61,0.515629298284671)
(0.62,0.513934700471096)
(0.63,0.512232030535226)
(0.64,0.510520271642956)
(0.65,0.508798393278044)
(0.66,0.507065349749677)
(0.67,0.505320078754026)
(0.68,0.503561499988837)
(0.69,0.501788513820122)
(0.7,0.5)
};
\draw [ultra thick, dashed] (axis cs:0.5,0.5)  -- (axis cs:0.7,0.5);
\legend{}
\end{axis}
\end{tikzpicture}

\vspace{0.25cm}
\begin{tikzpicture}
\begin{axis}[
name = axis3,
title = {$p = 0.75$},
xlabel = Value of $q$,
width = 0.5\textwidth,
height = 0.4\textwidth,
legend style = {at={(-0.1,-0.25)},anchor=north west,legend columns=3},
ymajorgrids = true,
xmin = 0.5,
xmax = 0.75,
ymin = 0.498,
ymax = 0.555,
]      

\addlegendentry{Catch-Up Rule$\quad$}
\addplot [blue, dotted, thick] coordinates {
(0.5,0.5440673828125)
(0.51,0.542229106189265)
(0.52,0.540401111523179)
(0.53,0.538582659882039)
(0.54,0.536773033597861)
(0.55,0.534971539606812)
(0.56,0.533177513137254)
(0.57,0.531390321737962)
(0.58,0.529609369638798)
(0.59,0.527834102436286)
(0.6,0.526064012096773)
(0.61,0.524298642270021)
(0.62,0.52253759390625)
(0.63,0.520780531169877)
(0.64,0.519027187643311)
(0.65,0.517277372814359)
(0.66,0.51553097884098)
(0.67,0.513787987587242)
(0.68,0.512048477924528)
(0.69,0.510312633292147)
(0.7,0.508580749511717)
(0.71,0.506853242849736)
(0.72,0.50513065832298)
(0.73,0.503413678241461)
(0.74,0.501703130983796)
(0.75,0.5)
};
\addlegendentry{Adjusted Catch-Up Rule$\quad$}
\addplot [red] coordinates {
(0.5,0.52850341796875)
(0.51,0.527203045850966)
(0.52,0.525926275986755)
(0.53,0.524671662000348)
(0.54,0.523437785310854)
(0.55,0.522223252993723)
(0.56,0.52102669572549)
(0.57,0.519846765809916)
(0.58,0.518682135283684)
(0.59,0.517531494099818)
(0.6,0.516393548387096)
(0.61,0.515267018783749)
(0.62,0.51415063884375)
(0.63,0.513043153514101)
(0.64,0.511943317681527)
(0.65,0.510849894787015)
(0.66,0.509761655506724)
(0.67,0.50867737649778)
(0.68,0.507595839207547)
(0.69,0.506515828744942)
(0.7,0.505436132812499)
(0.71,0.504355540697794)
(0.72,0.50327284232298)
(0.73,0.502186827351162)
(0.74,0.50109628434838)
(0.75,0.5)
};
\addlegendentry{Alternating ($ABBA$) Rule}
\addplot [ForestGreen, loosely dashed, thick] coordinates {
(0.5,0.543416341145833)
(0.51,0.5417262037701)
(0.52,0.540040225255628)
(0.53,0.538357348243028)
(0.54,0.536676526840624)
(0.55,0.534996723872949)
(0.56,0.533316908264052)
(0.57,0.531636052552565)
(0.58,0.529953130535551)
(0.59,0.528267115038188)
(0.6,0.526576975806451)
(0.61,0.524881677520056)
(0.62,0.523180177922916)
(0.63,0.521471426068529)
(0.64,0.519754360677706)
(0.65,0.51802790860615)
(0.66,0.516290983419461)
(0.67,0.514542484073185)
(0.68,0.512781293695597)
(0.69,0.511006278470963)
(0.7,0.509216286621093)
(0.71,0.507410147483021)
(0.72,0.505586670680744)
(0.73,0.503744645388969)
(0.74,0.501882839686883)
(0.75,0.5)
};
\draw [ultra thick, dashed] (axis cs:0.5,0.5)  -- (axis cs:0.75,0.5);
\end{axis}

\begin{axis}[
at = {(axis3.south east)},
xshift = 0.1\textwidth,
title = {$p = 0.8$},
xlabel = Value of $q$,
width = 0.5\textwidth,
height = 0.4\textwidth,
legend style = {at={(0,-0.2)},anchor=north west,legend columns=3},
ymajorgrids = true,
xmin = 0.5,
xmax = 0.8,
ymin = 0.498,
ymax = 0.555,
]      

\addlegendentry{Catch-Up Rule$\quad$}
\addplot [blue, dotted, thick] coordinates {
(0.5,0.554339999999999)
(0.51,0.55245345119556)
(0.52,0.550578258939448)
(0.53,0.548713514201494)
(0.54,0.546858322367039)
(0.55,0.545011805490196)
(0.56,0.543173105082879)
(0.57,0.541341385425044)
(0.58,0.539515837381973)
(0.59,0.537695682714915)
(0.6,0.535880178871795)
(0.61,0.534068624245116)
(0.62,0.532260363884574)
(0.63,0.530454795652269)
(0.64,0.528651376808805)
(0.65,0.526849631018867)
(0.66,0.525049155765266)
(0.67,0.523249630160701)
(0.68,0.521450823146908)
(0.69,0.51965260207105)
(0.7,0.517854941629629)
(0.71,0.516057933170377)
(0.72,0.514261794342951)
(0.73,0.51246687908946)
(0.74,0.510673687966158)
(0.75,0.508882878787878)
(0.76,0.507095277586996)
(0.77,0.505311889879008)
(0.78,0.503533912226981)
(0.79,0.50176274409738)
(0.8,0.5)
};
\addlegendentry{Adjusted Catch-Up Rule$\quad$}
\addplot [red] coordinates {
(0.5,0.537739999999999)
(0.51,0.536289210453938)
(0.52,0.534863879573048)
(0.53,0.533462588107087)
(0.54,0.532083946306869)
(0.55,0.530726590784314)
(0.56,0.529389181575679)
(0.57,0.528070399402457)
(0.58,0.526768943124559)
(0.59,0.525483527380601)
(0.6,0.524212880410256)
(0.61,0.522955742053779)
(0.62,0.521710861923976)
(0.63,0.52047699774603)
(0.64,0.519252913860732)
(0.65,0.518037379886791)
(0.66,0.516829169538066)
(0.67,0.515627059591605)
(0.68,0.514429829002621)
(0.69,0.513236258162518)
(0.7,0.512045128296295)
(0.71,0.510855220995723)
(0.72,0.509665317884786)
(0.73,0.508474200414014)
(0.74,0.507280649780395)
(0.75,0.506083446969696)
(0.76,0.504881372918057)
(0.77,0.503673208789864)
(0.78,0.502457736368978)
(0.79,0.501233738560451)
(0.8,0.5)
};
\addlegendentry{Alternating ($ABBA$) Rule}
\addplot [ForestGreen, loosely dashed, thick] coordinates {
(0.5,0.553459999999999)
(0.51,0.551743607981312)
(0.52,0.550033322536201)
(0.53,0.548327975973656)
(0.54,0.546626419247598)
(0.55,0.544927518169935)
(0.56,0.543230149890559)
(0.57,0.541533199637049)
(0.58,0.539835557706987)
(0.59,0.53813611670611)
(0.6,0.536433769025641)
(0.61,0.534727404552404)
(0.62,0.533015908605498)
(0.63,0.531298160093499)
(0.64,0.529573029886344)
(0.65,0.527839379396225)
(0.66,0.526096059361998)
(0.67,0.524341908831759)
(0.68,0.522575754338425)
(0.69,0.520796409263299)
(0.7,0.519002673382715)
(0.71,0.517193332593083)
(0.72,0.515367158809699)
(0.73,0.513522910034894)
(0.74,0.511659330591167)
(0.75,0.509775151515151)
(0.76,0.50786909110826)
(0.77,0.505939855640135)
(0.78,0.503986140200985)
(0.79,0.502006629699113)
(0.8,0.5)
};
\draw [ultra thick, dashed] (axis cs:0.5,0.5)  -- (axis cs:0.8,0.5);
\legend{}
\end{axis}
\end{tikzpicture}

\end{figure}


%% file: Figure3_fairness_condition.tex
\begin{figure}[ht!]
\centering
\caption{The fixed scoring probabilities in sudden death which guarantee that \\ the Adjusted Catch-Up Rule is fairer than the other penalty shootout designs}
\label{Fig3}

\begin{tikzpicture}
\begin{axis}[
xlabel = Value of $p$,
ylabel = Value of $q$,
width = 0.95\textwidth,
height = 0.5\textwidth,
legend style = {at={(0.2,-0.15)},anchor=north west,legend columns=6},
ymajorgrids = true,
max space between ticks=50,
xmin = 0.5,
xmax = 1,
ymin = -0.025,
ymax = 1,
]      

\addplot [name path=A, blue, very thick] coordinates {
(0.5,0.0293548629655568)
(0.51,0.0422275082602959)
(0.52,0.0549405974588224)
(0.53,0.0674970788296377)
(0.54,0.079899828448649)
(0.55,0.0921516523953994)
(0.56,0.104255288869607)
(0.57,0.116213410231363)
(0.58,0.128028624968195)
(0.59,0.139703479592011)
(0.6,0.151240460468851)
(0.61,0.162641995584189)
(0.62,0.173910456246421)
(0.63,0.185048158731057)
(0.64,0.196057365868001)
(0.65,0.20694028857421)
(0.66,0.217699087333915)
(0.67,0.228335873628476)
(0.68,0.238852711317871)
(0.69,0.249251617975701)
(0.7,0.259534566179549)
(0.71,0.269703484758399)
(0.72,0.279760259998791)
(0.73,0.28970673681129)
(0.74,0.299544719858785)
(0.75,0.309275974648064)
(0.76,0.318902228586056)
(0.77,0.328425172002066)
(0.78,0.337846459137274)
(0.79,0.347167709102712)
(0.8,0.356390506806885)
(0.81,0.36551640385416)
(0.82,0.374546919414978)
(0.83,0.383483541068926)
(0.84,0.392327725621647)
(0.85,0.401080899896533)
(0.86,0.409744461502099)
(0.87,0.41831977957591)
(0.88,0.426808195505893)
(0.89,0.435211023629828)
(0.9,0.443529551913787)
(0.91,0.451765042610256)
(0.92,0.459918732896645)
(0.93,0.467991835494873)
(0.94,0.475985539272656)
(0.95,0.483901009827163)
(0.96,0.491739390051598)
(0.97,0.49950180068532)
(0.98,0.507189340848043)
(0.99,0.514803088558653)
(1,0.522344101239149)
};

\addplot [name path=B, ForestGreen, very thick] coordinates {
(0.5,0.124513935959349)
(0.51,0.136334619722923)
(0.52,0.148038116818235)
(0.53,0.159626161279491)
(0.54,0.171100453097021)
(0.55,0.182462659048665)
(0.56,0.193714413506922)
(0.57,0.204857319222661)
(0.58,0.215892948086207)
(0.59,0.226822841866549)
(0.6,0.237648512929396)
(0.61,0.248371444934794)
(0.62,0.258993093514977)
(0.63,0.269514886933094)
(0.64,0.279938226723453)
(0.65,0.290264488313864)
(0.66,0.300495021630689)
(0.67,0.310631151687129)
(0.68,0.320674179155308)
(0.69,0.330625380922667)
(0.7,0.340486010633156)
(0.71,0.350257299213724)
(0.72,0.359940455386561)
(0.73,0.369536666167542)
(0.74,0.379047097351301)
(0.75,0.388472893983359)
(0.76,0.397815180819703)
(0.77,0.407075062774199)
(0.78,0.416253625354224)
(0.79,0.425351935084866)
(0.8,0.434371039922049)
(0.81,0.443311969654913)
(0.82,0.452175736297776)
(0.83,0.460963334471996)
(0.84,0.469675741778023)
(0.85,0.478313919157945)
(0.86,0.486878811248813)
(0.87,0.495371346726998)
(0.88,0.503792438643868)
(0.89,0.512142984753026)
(0.9,0.520423867829358)
(0.91,0.528635955980125)
(0.92,0.536780102948342)
(0.93,0.544857148408652)
(0.94,0.552867918255918)
(0.95,0.560813224886743)
(0.96,0.56869386747411)
(0.97,0.576510632235352)
(0.98,0.584264292693622)
(0.99,0.591955609933068)
(1,0.599585332847866)
};

\addplot [name path=C, black, very thick] coordinates {
(0.5,0.5)
(0.51,0.51)
(0.52,0.52)
(0.53,0.53)
(0.54,0.54)
(0.55,0.55)
(0.56,0.56)
(0.57,0.57)
(0.58,0.58)
(0.59,0.59)
(0.6,0.6)
(0.61,0.61)
(0.62,0.62)
(0.63,0.63)
(0.64,0.64)
(0.65,0.65)
(0.66,0.66)
(0.67,0.67)
(0.68,0.68)
(0.69,0.69)
(0.7,0.7)
(0.71,0.71)
(0.72,0.72)
(0.73,0.73)
(0.74,0.74)
(0.75,0.75)
(0.76,0.76)
(0.77,0.77)
(0.78,0.78)
(0.79,0.79)
(0.8,0.8)
(0.81,0.81)
(0.82,0.82)
(0.83,0.83)
(0.84,0.84)
(0.85,0.85)
(0.86,0.86)
(0.87,0.87)
(0.88,0.88)
(0.89,0.89)
(0.9,0.9)
(0.91,0.91)
(0.92,0.92)
(0.93,0.93)
(0.94,0.94)
(0.95,0.95)
(0.96,0.96)
(0.97,0.97)
(0.98,0.98)
(0.99,0.99)
(1,1)
};
\addplot[pattern = vertical lines, pattern color = blue] fill between[of = A and C];
\addplot[pattern = horizontal lines, pattern color = ForestGreen] fill between[of = B and C];
\end{axis}
\end{tikzpicture}

\vspace{-0.4cm}
\begin{center}
\begin{tikzpicture}
	\begin{customlegend}[legend columns=1,legend entries={Region where it is fairer than the Catch-Up Rule, Region where it is fairer than the Alternating ($ABBA$) Rule},
	legend image code/.code={\draw (0cm,-0.1cm) rectangle (0.6cm,0.2cm);}]
        \addlegendimage{color = blue, pattern = vertical lines, pattern color = blue, very thick}
        \addlegendimage{color = ForestGreen, pattern = horizontal lines, pattern color = ForestGreen, very thick} 
	\end{customlegend}
\end{tikzpicture}
\end{center}
\vspace{-0.5cm}
\end{figure}
